\documentclass[10pt,twocolumn,twoside]{IEEEtran}

\IEEEoverridecommandlockouts         
               
\usepackage{amsfonts}
\usepackage{amsmath}
\usepackage{amsthm}
\usepackage{amssymb}
\usepackage[pdftex]{color,graphicx}
\usepackage{enumerate}
\usepackage{bbm}
\usepackage{tikz}
\usepackage{subfigure}
\usepackage{pgfplots}
\usepackage{url}

\usepackage{caption}
\usepackage{wrapfig}

\usepackage{graphicx}

\usetikzlibrary{patterns,snakes}
\newtheorem{theorem}{Theorem}
\newtheorem{proposition}{Proposition}
\newtheorem{corollary}{Corollary}
\newtheorem{lemma}{Lemma}
\newtheorem{claim}{Claim}

\theoremstyle{definition}
\newtheorem{defn}{Definition}
\newtheorem{exmp}{Example}
\newtheorem{assumption}{Assumption}

\newcommand{\G}{{\mathcal{G}}}

\ifCLASSINFOpdf
\else
\fi
\usepackage{url}

\begin{document}
\title{Aggregate Fluctuations in Networks with Drift-Diffusion Models Driven by Stable Non-Gaussian Disturbances}
%
%
%
\author{Christoforos Somarakis and Nader Motee
\thanks{C. Somarakis and N. Motee are with the Department
of Mechanical Engineering \& Mechanics, Lehigh University, Bethlehem, PA, 18015, USA e-mail: \{csomarak,motee\}@lehigh.edu}
}
%


\maketitle

\begin{abstract}
The focus of this paper is to quantify measures of aggregate fluctuations for a class of consensus-seeking multi-agent networks subject to exogenous noise with $\alpha$-stable distributions. This type of noise is generated by a class of random measures with heavy-tailed  probability distributions. We define a cumulative scale parameter using scale parameters of probability distributions of the output variables, as a measure of aggregate fluctuation. Although this class of measures can be characterized implicitly in  closed-form in steady-state, finding their explicit forms in terms of network parameters is, in general, almost impossible. We obtain several tractable upper bounds in terms of Laplacian spectrum and statistics of the input noise. Our results suggest that relying on Gaussian-based optimal design algorithms will result in non-optimal solutions for networks that are driven by non-Gaussian noise inputs with $\alpha$-stable distributions. 
\end{abstract}

\begin{IEEEkeywords}
Consensus, Heavy-tailed noise, Performance Measures, $\alpha$-stable processes, $p$-norm.
\end{IEEEkeywords}
\section{Introduction}

%
%
%
%
\IEEEPARstart{T}he level of complexity in modern real-world networks can make them vulnerable to enviromental or structural disturbances often with severe, if not catastrophic, consequences. Recent crises in various sectors of our society show specific frailties of dynamical networks due to weaknesses in their structures, e.g., the air traffic congestion problem \cite{craig88}, power outages  \cite{poweroutage03}, the financial crisis of 2008 (see Ch. 17-18 in \cite{fouque13}) and other major disruptions. 

Thus the problem of performance and robustness in high dimensional networks is pivotal in designing inter-dependent systems that withstand negative effects of disturbances. Application areas include, but are not limited to, co-operative control of multi-agent systems, collaborative autonomy, transportation networks, power networks, metabolic pathways and financial networks (see for example \cite{6248170,Tahbaz13} and references therein).

Standard models of uncertainty in dynamic processes assume underlying probability distributions with well-defined first and second moments. A particularly prominent example is this of white noise, where the underlying distribution is Gaussian. Its main advantage is classic theory of stochastic differential equations (SDE) \cite{arnold74}, that provides clean and tractable results. This gives rise to Engineers to leverage Gaussian-induced sources of perturbation on networked control systems and design optimal structures that mitigate undesirable noise-related effects \cite{Bamieh12,yaserecc16,Siami16TACa,somyasnader17}. 

Despite elegant formulation, systems perturbed by purely Gaussian perturbations have attracted considerable criticism. The primary dispute relies on the claim that Gaussian approximations fail to model real-world uncertainty that occasionally exhibits susceptibility to large and unexpected events \cite{taleb2010black,mandelbrot2010mis,schoutens2003levy}. 

\subsection{Shortcomings of Gaussian Assumptions}
Systems perturbed by white noise generate stochastic processes that fluctuate around the expected (unperturbed) value in an amorphous yet highly regularized manner. The resulting dynamics essentially preserve the Gaussian nature of perturbations, along with its light-tailed property. Thus, there is no reasonable possibility for abrupt and outlying values to emerge, in other words faithful signs of large and unexpecrted fluctuations, or shocks, in the network.  As explained in \cite{taleb2010black}, shock events are ubiquitous in real world situations. Furthermore, they are identified as such, if they lie outside the realm of regular expectations, carry an extreme impact and have likelihood of happening. It is precisely the light-tailed property of Gaussian measures that hinders realistic possibility of shocks. To gain a better understanding, the qualitative difference in a solution trajectory perturbed by light-tailed and a solution trajectory perturbed  by heavy-tailed noise, is illustrated in Figure \ref{fig: orbit}.  
   

\subsection{Related Literature \& Contributions}
Mathematical models driven by non-Gaussian and heavy-tailed disturbances have been proposed in various disciplines \cite{shoutens05,duan15}. To the best of our knowledge, control community lacks relevant studies and results, with the exception of \cite{7084617}.
 
In this paper, we develop the theoretical framework of heavy-tailed consensus seeking networks. These are types of drift-diffusion stochastic differential equations, driven by $\alpha$-stable noise. The drift (deterministic) part is selected to be an average consensus protocol. This is the standard control law for asymptotic agreement over agents and it enjoys lasting interest in problems of cooperative dynamics, formation control, distributed computation and optimization \cite{Mesbahi_Egerstedt_2010}. The diffusion part consists of additive $\alpha$-stable random measures that model noise as exogenous disturbance on the unperturbed (drift) dynamics. 

The objective of our work to lay the groundwork analysis of this class of systems. It is found that the systemic (i.e., network-wide) response towards this class of noise is quantitatively and qualitatively different when compared to gaussian-induced systemic fluctuations. Furthermore we highlight the perplexed interplays between network topology and noise as a means to understand the manner with which shocks are propagated through the network, affecting its ability to remain in equilibrium. In particular, we introduce a measure to quantify aggregate flcutuation for networks driven by $\alpha$-stable noise. We derive an implicit formulation of the metric for system outputs and we explore its basic properties. Moreover, we highlight its connection with the $\mathcal H_2$-based performance measure for linear systems with white-noise inputs \cite{Siami16TACa} as well as with other $p$-metrics. Explicit expression of the systemic performance metric appears to be generally impossible. Exceptions are communication topologies with uniform, all-to-all connectivity or purely gaussian noise perturbations. For this reason we obtain tractable bounds of the performance metric which we believe to be useful in network design problems. Numerical examples are discussed and validate our theoretical results. We suggest that $\mathcal H_2$-based design algorithms are not only incompatible for the case of heavy-tailed  disturbances, but they also deliver sub-optimal topologies. It is acknowledged that the present work is an outgrowth of \cite{som_acc2018}. This version considers more general (i.e. not necessarily symmetric) $\alpha$-stable random measures, and it includes detailed proofs of technical results.

\begin{figure}\center
\includegraphics[scale=0.63]{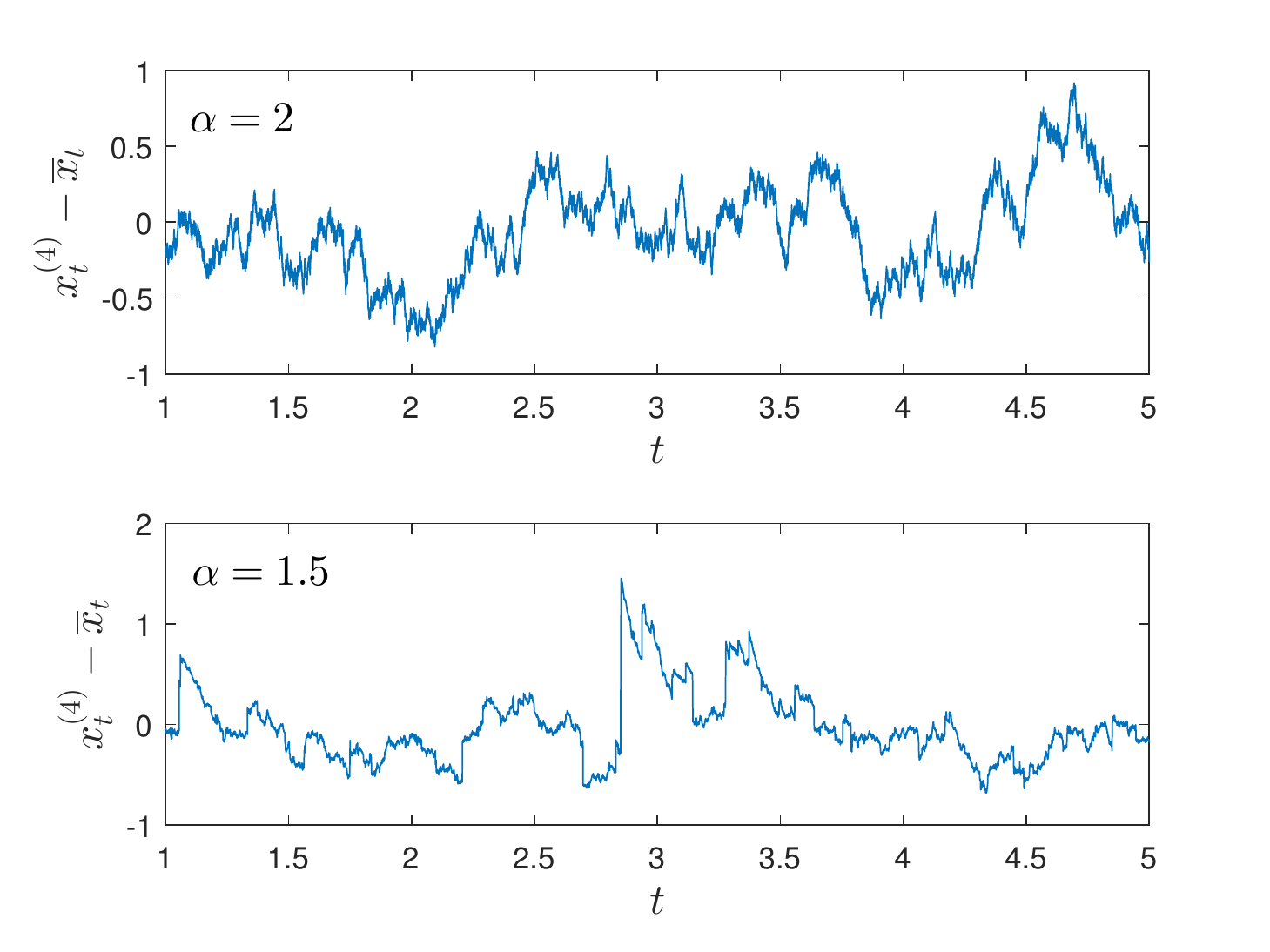}
\caption{Simulation of output dynamics of system \eqref{eq: model} for $n=7$ agents, in the face of white ($\alpha=2$) and heavy-tailed ($\alpha<2$) noise inputs. The latter type of perturbations results in dynamics with jumps that represent, more realistically, the effect of shocks on the nominal process.} \label{fig: orbit}
\end{figure}

\section{Preliminaries}\label{section: prelim}
\noindent 
By $\mathbb R^n$ we denote the $n$-dimensional Euclidean space, with elements $x=\big[x^{(1)},\dots,x^{(n)}\big]^T \in \mathbb R^n$. For any $x\in \mathbb R^n$, $\|x\|_p:=\sqrt[p]{\sum_{j}|x^{(j)}|^p}$, for $p>0$. The fundamental property on the equivalence of norms in $\mathbb R^n$:
\begin{equation*}
\| x\|_p \leq \| x\|_r \leq n^{\frac{1}{r}-\frac{1}{p}}\|x \|_p~~~\text{for}~p>r>0.
\end{equation*}

\noindent Given a probability space $(\Omega,\mathcal F, \mathbb P)$ we say that a random variable $z(\omega): \Omega\rightarrow \mathbb R$ follows a stable distribution, and we write $$z\sim S_{\alpha}(\sigma,\beta,\mu),$$ if there exist parameters $0<\alpha\leq 2$, $\sigma\geq 0$,  $-1\leq \beta \leq 1$ and $\mu \in \mathbb R$, such that its characteristic function is of the form:
\begin{equation*}
 \phi_z(\theta)=\mathbb E\big[e^{i\theta z}\big]=\text{exp}\big\{\sigma^\alpha\big (-|\theta|^\alpha+i\theta \omega(\theta,\alpha,\beta)\big)+i\mu \theta \big\}
\end{equation*} where $\omega(\theta,\alpha,\beta)$ stands for the function $$\omega(\theta,\alpha,\beta)=\begin{cases}
\beta |\theta|^{\alpha-1} \tan\frac{\pi \alpha}{2}, ~ \hspace{0.2in} \alpha\neq 1\\
-\beta \frac{2}{\pi} \ln |\theta| ,  ~ \hspace{0.46in} \alpha = 1.
\end{cases} $$
The parameter $\alpha$ is called the stability index of the distribution. Parameter $\alpha$ basically characterizes the impulsiveness (i.e. frequency and magnitude) of the shocks.  The parameter $\sigma$ is the scale of the distribution and it is closely related to the standard deviation: The larger the scale parameter is, the more spread out the distribution becomes. Parameter $\beta$ is the skeweness of the distribution, an indicator of asymmetry. Finally, $\mu$ is the shift of the distribution and it plays the role of the mean value\footnote{see Property 1.2.19 in \cite{samorodnitsky1994stable}.}. The following results summarize basic properties of stable random variables. They are drawn from \cite{samorodnitsky1994stable} and stated below as Propositions \ref{prop: propstrv}, \ref{prop: stableindep} and \ref{prop: stableintegralprop} to enhance readability and keep our manuscript self-contained.
 
%

\begin{proposition}\label{prop: propstrv} Let $z\sim S_{\alpha}(\sigma,\beta,\mu)$. It holds that:

\vspace{0.1in}
\noindent{1.} For any $a\in \mathbb R$, $z+a\sim S_{\alpha}(\sigma,\beta,\mu+a)$.

\vspace{0.1in}
\noindent{2.} For any $a\neq 0$ \begin{equation*}az\sim \begin{cases}S_{\alpha}\big(|a|\sigma,\text{sgn}(a)\beta,a\mu\big), & \alpha\neq 1 \\
 S_{\alpha}\big(|a|\sigma,\text{sgn}(a)\beta,a\mu-2a\ln(|a|)\sigma \beta\big), & \alpha=1
 \end{cases}
\end{equation*}

\vspace{0.1in}
\noindent{3.} If $\alpha<2$
\begin{equation*}
\mathbb E\big[|z|^p\big]\begin{cases}
<\infty, &~\text{for}~0<p<\alpha\\
=\infty, &~\text{for}~p\geq \alpha\\
\end{cases}
\end{equation*} In addition, if $\mu=0$, and $\beta=0$ only if $\alpha=1$, it holds   
$$\mathbb E[|z|^p]=c^p\sigma^p,$$ where $c=c(\alpha,\beta,p)=\big(\mathbb E[|z_0|^{p}]\big)^{\frac{1}{p}}$ for $z_0\sim S_{\alpha}(1,\beta,0)$. 
\end{proposition}

\vspace{0.1in}
\noindent A closed form expression of constant $c$ is reported in \cite{samorodnitsky1994stable}. Its most remarkable property is its limit at $\alpha$:  
\begin{equation*}
\lim_{p\rightarrow \alpha^-}c(\alpha,\beta,p)=\begin{cases}
+\infty, & \alpha <2\\
\sqrt{2}, & \alpha=2
\end{cases} \hspace{0.2in}~\text{for all}~\beta \in [-1,1].
\end{equation*}

\begin{proposition}\label{prop: stableindep}
Let $z_i\sim S_{\alpha}(\sigma_i,\beta_i,\mu_i),~i=1,2$ be independent. Then $$z_1+z_2\sim S_{\alpha}\bigg((\sigma_1^\alpha+\sigma_2^\alpha\big)^{\frac{1}{\alpha}},\frac{\beta_1\sigma_1^\alpha+\beta_2\sigma_2^\alpha}{\sigma_1^\alpha+\sigma_2^\alpha},\mu_1+\mu_2\bigg).$$
\end{proposition}

\vspace{0.1in}
\noindent The random variable $z\sim S_{\alpha}(\sigma,0,0)$ is called symmetric $\alpha$-stable, for which we write $z\sim S\alpha S$. Its characteristic function takes the form $$\phi_z(\theta)=e^{-\sigma^\alpha|\theta|^{\alpha}}.$$ 

\noindent A finite collection of $\alpha$-stable random variables $z_i\sim S_{\alpha}(\sigma_i,\beta_i,\mu_i),~i=1,\dots, d$ can form an $\alpha$-stable vector $z=\big[z^{(1)},\dots, z^{(d)}\big]^T$.

\noindent A scalar-valued stochastic process $\{z_t,~t\in [0,\infty]\}$ is stable if all its finite dimensional distributions are stable.  A nominal example is this of $\alpha$-stable L\'evy process $z=\{z_t,~t\geq 0\}$ with the properties:

\vspace{0.1in}
\noindent[1.] $z(0)=0$ a.s.

\vspace{0.1in}
\noindent[2.] $z$ attains independent increments

\vspace{0.1in}
\noindent[3.] $z_t-z_s\sim S_{\alpha}\big( (t-s)^{1/\alpha},\beta,0\big)$, for $0\leq s<t<\infty$.

\noindent A vector valued $\alpha$-stable random process $z=\{z_t\}_t$ with $z_t=\big[z_t^{(1)},\dots, z_t^{(d)}\big]^T$, $t\geq 0$  is a family of $\alpha$-stable vectors parametrized by $t$.

\vspace{0.1in}
\noindent\textit{Stable Integrals}. The building blocks of stable integrals are random measures.  Let $(\Omega, \mathcal F, P)$ be a probability space and $L^0(\Omega)$ the set of all real random variables defined on it. Let also $(B, \mathcal B , m)$ be a measure space. Take $\beta: B\rightarrow [-1,1]$ a measurable function and $\mathcal B_{0}\subset \mathcal B$ that contains sets of finite $m$-measure. 
\begin{defn}A set function $M: \mathcal B_0\rightarrow L^{0}(\Omega)$ is a random measure, if it satisfies the following properties:

\vspace{0.1in}
\noindent I. It is independently scattered, i.e. for any finite collection of disjoint sets $A_1,\dots,A_k \in \mathcal B_0$, the random variables $M(A_1),\dots,M(A_k)$ are independent. 

\vspace{0.1in}
\noindent II. It is $\sigma$-additive on $\mathcal B_0$. 


\vspace{0.1in}
\noindent III. For every $A\in \mathcal B_0$, $$M(A)\sim S_\alpha\bigg((m(A))^{1/\alpha},\frac{\int_A \beta(y)\, m(dy)}{m(A)},0\bigg)$$

\end{defn} 
\noindent 
The next example establishes an intimate connection between random measures and stable processes.
\begin{exmp}\label{exmpl: ex1}
Let $M$ be an $\alpha$-stable random measure on $\big( [0,\infty), \mathcal B \big)$ with $m(dx)=\frac{1}{\alpha}dx$ and constant skewness density $\beta$, $0\leq x < \infty$. The process $Z=\{ Z_t, t\geq 0\}$ defined through $Z_t=M([0,t]),~0\leq t<\infty$ is an $\alpha$-stable L\'evy motion.
\end{exmp}
\noindent The stable integral defined as $$I(f):=\int_{B}f(y)M(dy)$$ are taken over integrands that are members of \begin{equation}\label{eq: spacef}\begin{split} 
F_\alpha =\bigg\{ f\in \mathcal B: \int_B |f(y)|^{\alpha}\,m(dy)<\infty\bigg\}.
\end{split}
\end{equation}
\begin{proposition}\label{prop: stableintegralprop}The integral $I(f)$ attains the properties:

\vspace{0.1in}
\noindent{1.} $I(f)\sim S_{\alpha}(\sigma_f,\beta_f,\mu_f)$  with 
$ \sigma_f^{\alpha}=\int_B |f(x)|^\alpha\, m(dx)$, \\ $\beta_f=\frac{\beta}{\sigma_f}\int_B f(x)^{<\alpha>}\,m(dx)$, \text{and}  
$$\mu_f=\begin{cases} 0, & \alpha \neq 1 \\
-\frac{2}{\pi}\beta\int_{B}f(x)\ln |f(x)|\,m(dx), & \alpha =1.
\end{cases}$$
The notation $q^{<\alpha>}$ stands for \begin{equation*}
q^{<\alpha>}=\begin{cases}
|q|^\alpha &~\text{if}~q>0, \\
-|q|^{\alpha} & ~\text{if}~q<0.
\end{cases}
\end{equation*} 
\noindent{2.} $I(a_1f_1+a_2f_2)=a_1I(f_1)+a_2I(f_2)$,
for any $f_1,~f_2\in F_\alpha$, and constants  $a_1,~a_2\in \mathbb R$.
%
\end{proposition}

\begin{exmp}\label{exmp: stableintegral}
Let the $\alpha$-stable random measure $M$ of Example \ref{exmpl: ex1}. Then $f(x)=e^{-\lambda x}$, $\lambda>0$, clearly belongs to $F_\alpha$, $\alpha\in (0,2]$. For fixed $t>0$, the integral $$\int_0^{t}e^{-\lambda(t-s)}\,M(ds)~ \sim ~ S_\alpha(\sigma,\beta,\mu_f) $$ defines a stable process so that for   $\sigma^\alpha=\frac{1-e^{-\alpha \lambda t}}{\alpha^2 \lambda}$, $\beta_f=\beta$ and $\mu_f=-\frac{2}{\pi}\beta\big[te^{-\lambda t}-\frac{1}{\lambda}(1-e^{-\lambda t} )\big]$ if $\alpha=1$ and 0, otherwise.  
At $t \rightarrow \infty $ the stable integral converges, in distribution, to $S_\alpha\big(\frac{1}{\alpha \sqrt[\alpha]{\lambda}},\beta,0\big)$ for $\alpha\neq 1$ or $S_1\big(\frac{1}{\lambda},\beta,\frac{2}{\pi}\beta\big)$.
\end{exmp}

\vspace{0.1in}
\noindent {\it Algebraic Graph Theory.}
The vector of all ones is denoted by $\mathbf {1}$ and the $n \times n$ centering matrix is $$M_n := I_n - \frac{1}{n} \mathbf  1 \mathbf  1^T.$$ An undirected weighted graph $\mathcal{G}$ is defined by the triple $\mathcal{G}=(\mathcal V,\mathcal {E},w)$,  where $\mathcal V$ is the set of nodes of $\G$, $\mathcal{E}$ is the set of links of the graph,  and $w: \mathcal{E} \rightarrow \mathbb{R}_{+}$ is the weight function that maps each link to a non-negative scalar $a_{ij}$. The matrix $L=[l_{ij}]$ with $$l_{ij}=\begin{cases} -a_{ij},  & i\neq j \\ 
\sum_{j=1}^{n}a_{ij}, & i= j \end{cases}$$  is the Laplacian matrix of $\G$.  The following condition holds true throughout the paper.

\vspace{0.1cm} 
\begin{assumption}\label{assum0} The coupling graphs of all networks considered in this paper are simple, undirected, and  connected. \end{assumption}
\vspace{0.1cm}

\noindent A number of important consequences immediately follow. At first, $a_{ij}=a_{ji}$ for all $i,j\in \mathcal V$ that makes $L$ symmetric. Then its eigenvalues are real and they can be ordered as $$0=\lambda_1 <\lambda_2 \leq\dots\leq \lambda_{n}.$$ Furthermore, $L$ can be represented as $L=Q \Lambda Q^T $, where $\Lambda=\text{diag}(\lambda_1,\dots,\lambda_n)$  and 
$Q=[q_1~|~\dots~|~q_n]$ is a matrix the $i^{th}$ column of which is corresponds to the eigenvector associated with the eigenvalue $\lambda_i$ of $L$. Finally, $\{q_i\}_{i\in [n]}$ can be chosen to satisfy 
\[ q_i^T q_j=\left\{\begin{array}{ccc}
1 & \textrm{if} & i=j \\
0 & \textrm{if} & i\neq j. 
\end{array}\right.\] Under this normalization condition, the eigenvector of the smallest eigenvalue $\lambda_1=0$, takes the form $q_1=\frac{1}{\sqrt{n}} \mathbf  1$. For the sake of convenience, we define below a few graph laplacian related functions:
\begin{eqnarray}
f_{ij}(t)&=&\sum_{k=2}^n q_{ik}q_{jk}e^{-\lambda_k t} \label{eq: fij} \\
g(t) &=&\sum_{k=2}^n e^{-\lambda_k t} \label{eq: zeta}\\
G_{\alpha} &=& \int_0^{\infty}g^\alpha(s)\,ds  \label{eq: gfunction}
\end{eqnarray} 
where $\lambda_k$ the $k^{th}$ eigevnalues of $L=Q\Lambda Q^T$, $q_{ij}$ the $(i,j)$ element of $Q$, and $\alpha\in (0,2]$. Note that $f_{ij}$ clearly belong to $F_{\alpha}$. In addition, $|q_{ij}|\in [0,1]$ implies $|f_{ij}(t)|\leq g(t)$.
Additionally, we define
\begin{equation}\label{eq: kappaalpha}\begin{split}
\Lambda_{\alpha,p}^{(k)}=\Gamma^\frac{1}{p}(\alpha+1)\bigg[\sum_{m=2}^{k-1}&\frac{(\lambda_k-\lambda_m)^{\frac{\alpha}{p}}}{(\alpha \lambda_m)^{\frac{\alpha+1}{p}}}+\sum_{m=k+1}^{n}\frac{\big(\lambda_m-\lambda_k\big)^\frac{\alpha}{p}}{(\alpha\lambda_k)^{\frac{\alpha+1}{p}}}\bigg].
\end{split}
\end{equation} where $\Gamma(z)$ stands for the Gamma function. With a little abuse of notation, we define \begin{equation}\label{eq: kappaalphasum}\Lambda_{\alpha,p}=\sum_{k=2}^n \Lambda_{\alpha,p}^{(k)}.
\end{equation}

\section{Problem Statement}Consider a collection of $1,\dots,n$ autonomous agents, defined through the state $x^{(i)}\in \mathbb R$, $i=1,\dots,n$. The agents execute a consensus algorithm on a network with symmetric couplings to align their states. This alignment process is perturbed by $n$ noise sources powered by stable random motions. Every source is attached to node $i$ and it acts independently of the rest of the sources. This setting leads to the following system of stochastic differential equations:
\begin{equation}\label{eq: model}
dx_t=-L\,x_t\,dt+dz_t, \hspace{0.2in}  t>0 
\end{equation}  where $x_t=\big[x_t^{(1)},\dots,x^{(n)}_{t}\big]^T$ is the state vector, $L$ is the graph laplacian matrix that satisfies Assumption \ref{assum0}. Evidently, $dz_t=M(dt)$ is a multi-dimensional stable process under the next condition:
\begin{assumption}\label{assum: noise}
$dz_t=\big[M_1(dt),\dots,M_n(dt)\big]^T$ is a vector of $n$ independent random measures. For every $i=1,\dots,n$, the measure $M_i(dt)$ is defined on the measure space $([0,\infty),\mathcal B\big([0,\infty)\big),|\cdot|\big)$ such that
 $$M_i(t-s)\sim S_{\alpha}\big(|t-s|^{1/\alpha},\beta_i,0\big),~~\beta_i\in [-1,1]$$ is a random measure.
\end{assumption} 
The initial vector in system \eqref{eq: model}, $x_0=\big[x_0^{(1)},\dots,x_0^{(n)}\big]^T$, is arbitrary but fixed and it is chosen independently of $dz_t$. System \eqref{eq: model} is the differential form of a multi-dimensional generalized Ornstein-Uhlenbeck process, with integral representation
\begin{equation}\label{eq: integralmodel}
x_t=e^{-Lt}x_0+\int_{0}^t e^{-L(t-s)}dz_s
\end{equation}
Processes of this type have been studied in the past (see for example \cite{sato1983} and \cite{CIS-462223}) for $dz_s$ a generic stable measure and $-L$ being Hurwitz (i.e. $\lim_{t\rightarrow +\infty} e^{-L t}= O_{n\times n}$ ). 

The first objective of this paper is to study the fundamental properties of the solution of \eqref{eq: integralmodel}, define concepts of performance for \eqref{eq: model}, and calculate them explicitly, whenever possible. Otherwise we obtain faithful approximations and validate their efficiency. 


\section{ Output Signal Statistics}

Unlike the models discussed in \cite{sato1983} and \cite{CIS-462223}, $-L$ in \eqref{eq: model} is not Hurwitz. The interest in the study of consensus seeking systems is on observables that measure types of state differences. For example, we are interested in the relative agent displacement (i.e., $x^{(i)}-x^{(j)}$), or  agents' deviation from network average $\big($i.e.,  $x^{(i)}-\frac{1}{n}\sum_{j=1}^n x^{(j)}\big)$. For the latter case, stacking all the elements $i=1,\dots,n$ yields
\begin{equation}\label{eq: output}y=M_n x
\end{equation} where $M_n=I_n-\frac{1}{n}\mathbf 1\mathbf 1^T$ is the centering matrix. Applying this transformation to \eqref{eq: model} sets the marginal eigenvalue unobservable so that noise-free output is asymptotically stable. Also, the noisy output process $y=\{ y_t = M_n\,x_t, t\geq 0\}$ enjoys a number of remarkable properties summarized below. 

\begin{theorem}\label{prop: ydist} Under Assumptions \ref{assum0} and \ref{assum: noise}, the process $y=\{y_t,~t\geq 0\}$ in \eqref{eq: output} generated by $x=\{x_t,~t\geq 0\}$ to be the realization of \eqref{eq: integralmodel}, satisfies:
\begin{equation}\label{eq: outputdynamics}
y_t=Q \Phi(t) Q^T y_0+\int_{0}^{t}Q \Phi(t-s) Q^T dz_s,
\end{equation} where $$\Phi(t)=\text{Diag}\big[0,e^{-\lambda_2 t},\dots, e^{-\lambda_n t}\big]$$ and $\{\lambda_i\}_{i=2}^n$ the eigenvalues of $L$. For every fixed $t$, $y_t$ is a stable vector, with the $i^{th}$ element $y_t^{(i)}$ a stable random variable with $t$-dependent distribution parameters. As $t\rightarrow \infty$, the $l^{th}$ element of $\overline{y}=\lim_{t} y_t$, is distributed as  
$$ \overline{y}^{(l)}~\sim~S_{\alpha}\big(\sigma_l,\beta_l,\mu_l\big)$$ where 
\begin{equation}\label{eq: stableparameters}
\begin{split}
\sigma_{l}^\alpha&=\sum_{j=1}^n\sigma_{lj}^\alpha\\
\beta_{l}&=\frac{1}{\alpha}\frac{\sum_{j}\beta_j\sigma_{lj}^\alpha \int_0^{\infty}f_{lj}(s)^{<\alpha>}\,ds}{\sigma_j^\alpha}\\
\mu_l&=\begin{cases}
0, & \alpha\neq 1\\
-\frac{2}{\pi}\sum_{j}\beta_j\int_{0}^\infty f_{lj}(s)\ln |f_{lj}(s) |\,ds , & \alpha =1
\end{cases}
\end{split}
\end{equation} with 
\begin{equation}
\label{eq: sigmaij}
\sigma_{lj}^\alpha=\frac{1}{\alpha}\int_0^{\infty}|f_{lj}(s)|^\alpha\,ds.
\end{equation}
%
\end{theorem}
\begin{proof}
For the first part of the proof we observe that $M_n$ can be expressed as $M_n=Q E Q^T$, where $Q$ is the eigenvector matrix of $L$, and $E$ the $n\times n$ diagonal matrix with structure $E=\text{Diag}[0,1,\dots, 1]$. For $y_t=M_n x_t$, we have
\begin{equation*}\begin{split}
y_t&= QEQ^T Q e^{-\Lambda t} Q^T x_0+QEQ^T \int_{0}^t Q e^{-\Lambda(t-s)}Q^T dz_s\\
&= Q \Phi(t) Q^T x_0+\int_{0}^t Q\Phi(t-s) Q^T dz_s \\
&=Q \Phi(t)\big(EQ^T x_0\big)+\int_{0}^t Q\Phi(t-s) Q^T dz_s \\
&=Q \Phi(t) Q^T y_0+\int_{0}^t Q\Phi(t-s) Q^T dz_s.
\end{split} 
\end{equation*} The second step is due to the linearity of the integral operator in Proposition \ref{prop: stableintegralprop}. The $l^{th}$ element of $y_t$, equals\footnote{We consider equality in the sense of distribution, when we refer to stochastic processes.}
\begin{equation*}
y_t^{(l)}=\sum_{j=1}^n f_{lj}(t)y_0^{(j)}+\sum_{j=1}^{n}\int_{0}^t f_{lj}(t-s)\,M_j(ds),
\end{equation*}  where $f_{ij}(t)$ as in \eqref{eq: fij}. \noindent In other words, $y_t^{(l)}$ is equal to a transient constant term plus the sum of $n$ independent $\alpha$-stable integrals, each of which involves an $m$-measurable function.  From Proposition \ref{prop: stableintegralprop}, the $j^{th}$ stable integral $$\int_{0}^t f_{lj}(t-s)\,M_j(ds) \sim S_{\alpha}\big(\sigma_{lj}(t),\beta_{lj}(t),\mu_{lj}(t)\big)$$ with 

$\sigma_{lj}(t)^\alpha=\frac{1}{\alpha}\int_0^{t}|f_{lj}(s)|^\alpha\,ds$, $\beta_{lj}(t)=\frac{\beta_j}{\alpha} \frac{\int_{0}^t f_{lj}(s)^{<\alpha>}\,ds}{\sigma_{lj}(t)} $, and $\mu_{lj}(t)\equiv 0$ if $\alpha\neq 1$ and  $\mu_{lj}(t)=-\frac{2\beta_j}{\pi}\int_0^t f_{lj}(s)\ln|f_{lj}(s)|\,ds$, otherwise. An inductive application of Proposition \ref{prop: stableindep} implies that the sum of $n$ independent stable integrals, is a stable random variable: $$\sum_{j=1}^{n}\int_{0}^t f_{lj}(t-s)\,M_j(ds)\sim S_{\alpha}\big(\sigma_l(t),\beta_l(t),\mu_l(t)\big)$$ with $\sigma_{l}^\alpha(t)=\sum_{j} \sigma_{lj}^\alpha(t)$, $\beta_l(t)=\frac{1}{\alpha}\frac{\sum_{j}\beta_{lj}(t)}{\sum_{j}\sigma_{lj}^\alpha(t)}$, and $\mu_{l}(t)=\sum_{j}\mu_{lj}(t)$. 
The result follows immediately after taking the limit in $t$. 
\end{proof}

Theorem \ref{prop: ydist} explains that the distance of agents from network average follows a well-defined stable distribution for all times. 
It is remarked that the network topology affects the spread of the distribution, the symmetry and if $\alpha=1$, also the shift parameter. Network topology does not, however, impact stability index $\alpha$. We conclude that the deterministic process (in our case the network topology) cannot affect the tail of the distribution. The impulsiveness and frequency of the shocks will continue to affect the system regardless of its structure. The network can, to some extend, handle its ability to remain rigid in the face of these shocks.  


Another observation due at this point, is that distribution parameters, although valuable, are quite difficult to be expressed in closed form.  Unfortunately, $\alpha$-stable processes are not famous for yielding elegant formulas, especially for multi-dimensional systems \cite{samorodnitsky1994stable}. In an interesting turn of events, there is a remarkable exception to this major difficulty for linear consensus systems.
\begin{corollary}\label{cor: complete}
If for the graph laplacian spectrum, it holds that $\lambda_2=\lambda_n=:\lambda$ then for any $t\geq 0$
\begin{equation*}
y_t^{(l)}~\sim~S_{\alpha}\big(\sigma_l(t),\beta_l(t),\mu_l(t)\big)
\end{equation*} with
\begin{equation*}
\begin{split}
\sigma_l(t)&=\frac{(n-1)+(n-1)^\alpha}{n^\alpha \alpha^2\lambda}\big(1-e^{-\alpha\lambda t}\big)\\
\beta_l(t)&=\frac{(1-e^{-\alpha \lambda t})\big(\beta_i (n-1)^\alpha - \sum_{j\neq i}\beta_j\big)}{n^\alpha \alpha^3 \lambda\big[(n+1)(1+(n-1)^{\alpha-1}) \big]} \\
\mu_l(t)&=\begin{cases}
0, & \alpha \neq 1 \\
2\frac{\lambda^{-1}(1-e^{-\lambda t})-te^{-\lambda t}}{n\pi}\big((n-1)\beta_l - \sum_{j\neq l}\beta_j \big),&\alpha=1.
\end{cases}
\end{split}
\end{equation*}
\end{corollary}
\begin{proof}
Condition $\lambda_2=\lambda_n$ implies $\lambda_2=\lambda_3=\dots=\lambda_n=\lambda>0$. Also, by virtue of symmetry on $L$ the matrix $Q$ consists of unit length mutually orthogonal columns as well as rows. In view of $q_1=\frac{1}{\sqrt{n}}\mathbf 1$, it is straightforward 
\begin{equation*}
f_{ij}(t)=\begin{cases}
-\frac{1}{n}e^{-\lambda t}, & i\neq j \\
\frac{n-1}{n}e^{-\lambda t}, & i=j.
\end{cases} 
\end{equation*}Consequently,
\begin{equation*}
\sigma_{ij}^\alpha(t)=\begin{cases} 
\frac{1}{n^\alpha\alpha^2\lambda}\big(1-e^{-\alpha\lambda t}\big), & i\neq j\\
\frac{(n-1)^\alpha}{n^{\alpha}\alpha^2\lambda}\big(1-e^{-\alpha\lambda t}\big), & i=j.
\end{cases}
\end{equation*} The result follows by straightforward algebra.
\end{proof}
Canonical example of a graph with identical non-zero laplacian eigenvalues is the complete graph with uniform coupling weights\footnote{Also called complete topological graph.}. Although Corollary \ref{cor: complete} assumes such a special case of connectivity, one can make a few network related significant remarks. Corollary \ref{cor: complete} suggests that for fixed number of agents and increased connectivity (i.e. $\lambda>> 1$) the scale, the skew and the shift of the distribution deteriorate as $\mathcal O(\lambda^{-1})$. On the other hand, growth of network with fixed communication weights  (i.e. $n>>1$) reveals essentially $\alpha$-dependent behavior. To see this let us for a moment focus on on symmetric $\alpha$-stable noise (i.e. $\beta=\mu=0$). In such case, scale $\sigma_l$ grows as $\mathcal O(n^{1-\alpha})$ when noise sources do not attain finite first moments  (i.e. $\alpha $ in the range of $(0,1)$). On the other hand, scale converges to $\frac{1-e^{-\alpha \lambda t}}{\alpha^2 \lambda} $ if noise has finite first moments (i.e. $\alpha$ in the range of $(1,2]$). The direct implication of Corollary \ref{cor: complete} is that large-scale networks (in terms of number of nodes) may exhibit higher deviations than small-scale networks, when additive noise induce shocks of increased frequency and impact (i.e. with infinite expectation). The situation is reversed when noise is less impulsive (i.e. $\alpha\in [1,2]$). 

\section{Measures of Aggregate Deviations}
For stability index $\alpha=2$, we recover the Gaussian-based stochastic behavior of $y=\{y_t,~t\geq 0\}$. The statistical properties of interest are rendered from their first and second moments, both of which are well-defined and asymptotically constant. For networks like \eqref{eq: model} researchers focus on the aggregate variability of the output, $\mathbb E\big[\| y_t \|^2\big]$, in order to measure its behavior in the face of noise.  As Proposition \ref{prop: propstrv} explains, this is not possible for stable noise with $\alpha<2$. This poses the question on how could one quantify the impact of noise to a dynamical system hit by heavy-tailed noise. One answer could be the sum of scales in a $\alpha$-stable vector.
%
%
\vspace{0.1in}
\begin{defn}\label{defn: vectorscale}
The cumulative scale of an $\alpha$-stable vector $y=[y_1,\dots,y_m]^T$ is defined to be
\begin{equation*}
\Sigma_\alpha(y)=\|\sigma\|_{\alpha}^{\alpha}=\sum_{l=1}^{m}\sigma_l^\alpha 
\end{equation*} where $\sigma=[\sigma_1,\dots,\sigma_m]^T$ and $\sigma_l$ is the scale parameter of the $l^{th}$ element of $y$.
\end{defn} 

\noindent For $\overline{y}$, the long term output vector of \eqref{eq: outputdynamics}, $\Sigma_\alpha(\overline{y})$ can be trivially expressed in terms of the stable integrals \eqref{eq: sigmaij}:

\begin{defn}The steady-state aggregate fluctuations of output dynamics \eqref{eq: outputdynamics} are defined to be
\begin{equation}\label{eq: Sigmaalpha}
\Sigma_\alpha(\overline{y})=\frac{1}{\alpha}\sum_{i=1}^n \sum_{j=1}^n \int_0^\infty |f_{ij}(t)|^\alpha \,dt.
\end{equation}
\end{defn} Evidently, $\Sigma_{\alpha}(y)$ for $y$ as in \eqref{eq: output}, is a measure of steady-state dispersion of agents around the moving average. The larger the $\Sigma_{\alpha}(\overline{y})$, the more impulsive and magnified the fluctuation of the agents around the moving average is. The spectral functions $f_{ij},~i,j \in \mathcal V$ are as in \eqref{eq: fij} and represent the network contribution in the form of the steady-state distribution of $\overline{y}$. In other words, $\sigma_{ij}$ contains all the information that is for primary interest to a network analyst.  The next result asserts that $\Sigma_{\alpha}(\overline{y})$ decreases with $\alpha$.

\begin{proposition}\label{prop: monotonicity}
Assume the network dynamics of \eqref{eq: model} with the output process \eqref{eq: output}. Then
\begin{equation*}
\frac{\partial}{\partial \alpha}\Sigma_{\alpha}(\overline{y})<0.
\end{equation*}
\end{proposition}
\begin{proof}
From the definition of $\Sigma_\alpha$ in \eqref{eq: Sigmaalpha} it suffices to prove $\frac{\partial}{\partial \alpha}\sigma_{ij}^\alpha<0$. This is equivalent to
\begin{equation*}
-\frac{1}{\alpha^2}\int_0^{\infty}|f_{ij}(t)|^\alpha\,dt+\frac{1}{\alpha}\int_0^{\infty}\ln \big(|f_{ij}(t)|\big) |f_{ij}(t)|^\alpha\,dt<0.
\end{equation*} The latter condition is true if,
$\ln \big(|f_{ij}(t)|\big)<\frac{1}{\alpha}$. This is in turn equivalent to
$|f_{ij}(t)|<e^{1/\alpha}$. The latter inequality is, however, true in view of 
\begin{equation*}\begin{split}|f_{ij}(t)|&\leq e^{-\lambda_2 t} \sum_{k}|q_{ik}| |q_{jk}|\\
&\leq e^{-\lambda_2 t} \sqrt{\sum_{k}|q_{ik}|^2} \sqrt{\sum_{k}|q_{jk}|^2}<1\end{split}\end{equation*}
by virtue of the Cauchy-Schwarz inequality and the properties of normalized Laplacian eigenvectors.
\end{proof}
\noindent In conclusion, the more impulsive the noise, the more the states of the network are prone to exhibit large and frequent deviations. For $\alpha=2$, Assumption \ref{assum0} and Property \ref{prop: propstrv} yield
\begin{equation}\label{eq: h2norm}
\begin{split}\hspace{-0.1in}
\Sigma_{2}(\overline{y})&=\frac{1}{2}\sum_{i,j}\int_0^\infty f_{ij}^{2}(t)\,dt=\frac{1}{2}\sum_{k=2}^n\frac{1}{2\lambda_k}=\frac{1}{2}\mathbb E\big[\|\overline{y}\|_2^2\big],
\end{split}
\end{equation} where $\lambda_k$ are the eigenvalues of $L$, and the last step is in view of Property 3 of Proposition \ref{prop: propstrv}. $\Sigma_2$ is intimately related to the cumulative variance of the output $\overline{y}$ of system \eqref{eq: model}, i.e. the $\mathcal H_2$-norm of the consensus network; a central measure of performance in stochastically driven dynamical systems \cite{Siami16TACa}.
The Gaussian case is unique in its kind, in the sense that leads to a closed form expressions of $\Sigma_2$. Clearly, the calculation above is not correct when $\alpha<2$. It seems that no other value of the stability index offers this elegance, with the exception of complete topological graph, that can be directly calculated using Corollary \ref{cor: complete} as:
\begin{equation}\label{eq: completescale}
\Sigma_{\alpha}(\overline{y})=\frac{(n-1)\big(1+(n-1)^{\alpha-1}\big)}{\alpha^2 n^{\alpha-1} \lambda}
\end{equation} where $\alpha\in (0,2]$ and $\lambda:=\lambda_2=\lambda_3,\dots=\lambda_n>0$.

%
\section{Spectral Based Bounds}
\noindent Stable integrals as in \eqref{eq: sigmaij} are indicative of the extent to which $\Sigma_\alpha$ can be calculated in closed form. With the exception of \eqref{eq: completescale}, one may need to rely on estimates of  aggregate steady-state scale $\Sigma_\alpha(\overline{y})$ for dynamical networks such as \eqref{eq: model}. The purpose of this section is to elaborate on \eqref{eq: Sigmaalpha} and establish upper estimates on $\Sigma_{\alpha}$. It is desirable to express these estimates as explicit functions of the eigenstructure of $L$, given the feature of noise. Our strategy is to construct estimates that become sharp as $\lambda_2\uparrow \lambda_n$ and/or as $\alpha \uparrow 2$, so as to resonate with the two extreme cases of connectivity and noise. 

\vspace{0.1in}

\begin{theorem}\label{thm: main1} Assume network \eqref{eq: model} with Assumptions \ref{assum0} and \ref{assum: noise} to hold and the stability parameter $\alpha \in (0,2]$ and consider the output vector-valued process $y=\{y_t,~t\geq 0\}$ from \eqref{eq: outputdynamics}. The following estimates on $\Sigma_{\alpha}(\overline{y})$ hold:

\noindent If $\alpha \in (0,1]$,

\begin{equation*}
\Sigma_\alpha(\overline{y}) \leq c_1 \sum_{k=2}^{n}\|q_k\|_{\alpha}^{2\alpha}\Lambda_{\alpha,1}^{(k)}+ c_2 \,G_{\alpha},
\end{equation*} for $c_1$, $c_2$ the constants
\begin{equation*}
c_1=\frac{1}{\alpha (n-1)^\alpha}\hspace{0.1in}\text{and}\hspace{0.1in} c_2=\frac{1+(n-1)^{1-\alpha}}{\alpha n^{\alpha -1}}.
\end{equation*}

\noindent If $\alpha \in [1,2]$,

\begin{equation*}\begin{split}
\Sigma_{\alpha}(\overline{y})\leq \min\bigg\{&d_1 \Lambda_{\alpha,\alpha}^{\alpha-1}\sum_{k=2}^n \| q_k \|_\alpha^{2\alpha}\Lambda_{\alpha,\alpha}^{(k)}+ d_2 G_\alpha, \\
&\hspace{0.33in}d_3 \Lambda_{\alpha,\alpha}^{\alpha-1} \sum_{k=2}^n \| q_k \|_\alpha^{2\alpha}\Lambda_{\alpha,\alpha}^{(k)}+d_4 G_\alpha\bigg\}
\end{split}
\end{equation*}  for $d_1,d_2,d_3,d_4$ defined to be 
$$d_1=\frac{2^{\alpha-1}}{\alpha (n-1)^\alpha}, \hspace{0.2in} d_2=\frac{2^{\alpha-1}n^{1-\alpha}}{\alpha}(1+(n-1)^{1-\alpha}),$$
$$ d_3=\frac{1}{\alpha (n-1)^\alpha},\hspace{0.3in}d_4=\frac{(1+(n-1)^{1-\alpha})(1+\alpha\Lambda_{\alpha,\alpha}^{\alpha-1})}{n^{\alpha-1}(n-1)^{-\alpha}}$$ The sum is taken over the non-zero eigenvalues $\lambda_k$ of the graph Laplacian $L$ with $q_{k}$ to be the $k^{th}$ eigenvector that corresponds to the  $\lambda_k$ eigenvalue.
Also, $\Lambda_{\alpha,\alpha}^{(k)}$ as in \eqref{eq: kappaalpha}, $\Lambda_{\alpha,\alpha}$ as in \eqref{eq: kappaalphasum} and $G_\alpha$ as in \eqref{eq: gfunction}.
\end{theorem}
\begin{proof} From Definition \ref{defn: vectorscale} 
\begin{equation}\label{eq: sigmaaexpansion}
\Sigma_\alpha=\sum_{i=1}^n \sigma_i^\alpha=\sum_{i=1}^n \sum_{j=1}^n \sigma_{ij}^\alpha=\sum_{j\neq i=1}^n\sigma_{ij}^\alpha+\sum_{i=1}^n\sigma_{ii}^\alpha
\end{equation} as it occurs from Proposition \ref{prop: propstrv} and Theorem \ref{prop: ydist}. The following Claims are central estimates of $\sigma_{ij}^\alpha$ for $\alpha \in (0,1]$ and 
$\alpha \in [1,2]$ respectively. Their proof is put in the Appendix.

\noindent We begin with the case $\alpha\in (0,1]$. 
\begin{claim}\label{lem: sigmaijestimates} If $\alpha\in (0,1]$, the following estimates hold true:
\begin{equation*}
\sigma_{ij}^{\alpha}\leq \begin{cases} c_1\sum_{k}|q_{ik}|^\alpha |q_{jk}|^{\alpha}\Lambda_{\alpha,1}(k)+\overline{c}_2 \,G_\alpha,& i\neq j \\ 
 \\
c_1\sum_{k}|q_{ik}|^{2\alpha}\Lambda_{\alpha,1}(k)+\underline{c}_2\,G_\alpha , & i=j
\end{cases}
\end{equation*} where $\overline{c}_2=\frac{1}{n^\alpha \alpha (n-1)^{\alpha}}$ and $\underline{c}_2=\frac{1}{\alpha n^{\alpha}}$.
\end{claim}
\noindent the first part of the result follows by direct application of the bounds of $\sigma_{ij}^\alpha$ of Claim \ref{lem: sigmaijestimates} in \eqref{eq: sigmaaexpansion}. We continue with the case $\alpha \in [1,2)$. We make a similar claim on upper bounds of $\sigma_{ij}^\alpha$.
\begin{claim}\label{lem: sigmaijestimates2} If $\alpha \in [1,2]$ then, for $i\neq j$, either
\begin{equation*}
\sigma_{ij}^\alpha\leq  \frac{2^{\alpha-1}}{\alpha(n-1)^\alpha}\bigg[\Lambda_{\alpha,\alpha}^{\alpha-1}\sum_k |q_{ik}|^\alpha |q_{jk}|^\alpha \Lambda_{\alpha,\alpha}(k)+\frac{G_\alpha}{n^\alpha}\bigg]
\end{equation*} 
or 
\begin{equation*}\begin{split}
\sigma_{ij}^\alpha\leq & \frac{1}{\alpha(n-1)^\alpha}\bigg[\Lambda_{\alpha,\alpha}^{\alpha-1}\sum_k |q_{ik}|^\alpha |q_{jk}|^\alpha \Lambda_{\alpha,\alpha}(k)+\\
&\hspace{1.8in}+\frac{G_\alpha}{n^\alpha}\big(1+\alpha \Lambda_{\alpha,\alpha}^{\alpha-1}\big)\bigg]
\end{split}
\end{equation*} Also, for $i=j$, either

\begin{equation*}
\sigma_{ii}^\alpha\leq  \frac{2^{\alpha-1}}{\alpha(n-1)^\alpha}\bigg[\Lambda_{\alpha,\alpha}^{\alpha-1}\sum_k |q_{ik}|^{2\alpha} \Lambda_{\alpha,\alpha}(k)+\frac{(n-1)^\alpha}{n^\alpha}G_\alpha\bigg]
\end{equation*} 
or
\begin{equation*}\begin{split}
\sigma_{ii}^\alpha&\leq \frac{1}{\alpha(n-1)^\alpha}\bigg[\Lambda_{\alpha,\alpha}^{\alpha-1}\sum_k |q_{ik}|^{2\alpha} \Lambda_{\alpha,\alpha}(k)\\
&\hspace{1.4in}+\frac{(n-1)^\alpha}{n^\alpha} G_\alpha\big(1+\alpha \Lambda_{\alpha,\alpha}^{\alpha-1}\big)\bigg]
\end{split}
\end{equation*} 
\end{claim}
\noindent The second part of the result follows in a similar manner to the first. 
\end{proof}

A worth mentioning technical remark that occurs from Theorem \ref{thm: main1} is the technical distinction between estimates obtained with noise sources for finite first moments, i.e. $\alpha \in (1,2]$, and estimates for noise with infinite first moment, i.e. $\alpha \in (0,1]$. In either case, bounds are generally constituted  of two terms: The first term equals the weighted sum of the $\alpha$-norm of the $n-1$ eigenvectors of $L$. The weight of the $k^{th}$ term in this sum is an eigenvalue-based function that essentially measures the deviation of the $k^{th}$ eigenvalue with respect to the rest $n-2$. The second term effectively involves the sum of the inverse non-zero eigenvalues of $L$ that it is expressed in integral form. One can sacrifice additional sharpness and use the simple bound $G_\alpha \leq \frac{(n-1)}{\lambda_2}$.

We remark that $\lambda_2 \uparrow \lambda_n$ implies $\Lambda_{\alpha,\alpha} \downarrow 0$. The estimates of Theorem \ref{thm: main1} coincide with the exact value of $\Sigma_\alpha(\overline{y})$ in \eqref{eq: completescale}. However, for $\alpha=2$, the estimates in Theorem \ref{thm: main1} do not match with the value in \eqref{eq: h2norm}. This non-negligible discrepancy motivates the additional upper bound of $\Sigma_{\alpha}(\overline{y})$.
\subsection{Estimates near $\alpha=2$.}
\noindent Together with Theorem \ref{thm: main1} we propose a different, yet particularly simple approach, in establishing estimates of $\Sigma_\alpha$, via a harmless perturbation of the scale parameter from the Gaussian case $\alpha=2$.
\vspace{0.1in}
\begin{theorem}\label{thm: main2} Assume network \eqref{eq: model} with Assumptions \ref{assum0} and \ref{assum: noise} to hold and the stability parameter $\alpha \in (1,2]$. Let the output vector-valued process $y=\{y_t,~t\geq 0\}$ as in \eqref{eq: outputdynamics}. Then,
\begin{equation*}\begin{split}
&\Sigma_\alpha(\overline{y})\leq \frac{1}{\alpha}\sum_{k=2}^{n}\frac{1}{2\lambda_k} +\frac{1}{\alpha}\int_{\alpha}^2 \int_{0}^\infty  n^{2-w} g^{w}(s)\big|\ln g(s)\big|\,dsdw.
\end{split}
\end{equation*} where $g(t)$ is as in \eqref{eq: zeta}.
\end{theorem} 
\begin{proof}Rewrite $\sigma_{ij}^\alpha$ as a harmless perturbation of $\frac{1}{\alpha}\int_0^\infty |f_{ij}(s)|^2\,ds$, as follows:
\begin{equation*}\begin{split}
&\sigma_{ij}^{\alpha}=\frac{1}{\alpha}\int_{0}^\infty |f_{ij}(s)|^\alpha\,ds\\
&=\frac{1}{\alpha}\int_{0}^\infty |f_{ij}(s)|^2 \,ds+\frac{1}{\alpha}\int_{0}^\infty\int_{2}^\alpha \ln(|f_{ij}(s)|)|f_{ij}(s)|^{w} \,dwds\\
&=\frac{1}{\alpha} A_{ij}+\frac{1}{\alpha} B_{ij}
\end{split}
\end{equation*}
The integrand in $A_{ij}$ reads \begin{equation*}
\sum_{k=2}^{n}q_{ik}^2 e^{-2\lambda_k s}+\sum_{k_1\neq k_2=2}^{n}q_{ik_1}q_{jk_1}q_{ik_2}q_{jk_2} e^{-(\lambda_{k_1}+\lambda_{k_2}) s}
\end{equation*}
so that summing over $j$ \begin{equation}\label{eq: I1}\sum_{k=2}^{n}q_{ik}^2 e^{-2\lambda_k s},\end{equation} from the eigenvectors property $\sum_{j=1}^{n}q_{jk}\equiv 0~\forall~k>1$. We proceed with $B_{ij}$. 
From the convexity of $r(t)=|t|^w$ for every $w\in [\alpha,2]$, Lemma \ref{lem: jensen} yields
\begin{equation*}
|f_{ij}(s)|^{w}\leq g^{w-1}(s)\sum_{k=2}^{n} |q_{ik}|^w |q_{jk}|^w e^{-\lambda_k s}
\end{equation*} and $|f_{ij}(s)|\leq g(s) \Rightarrow |\ln(|f_{ij}(s)|)|\leq \big| \ln g(s)\big|$  where $g(t)=\sum_{k=2}^n e^{-\lambda_k t}$.\begin{equation*}\begin{split}
B_{ij}&\leq \sum_{k=2}^{n}\int_{0}^\infty\int_{\alpha}^2 |q_{ik}|^w |q_{jk}|^w  g^{w-1}(s)\big |\ln g(s)\big|e^{-\lambda_k s} \,dwds\\
\end{split}
\end{equation*}
Taking the double sum over $i$ and $j$, in view of $\sum_{i}q_{ik}^2\equiv 1$,
\begin{equation*}\begin{split}
&\Sigma_\alpha=\sum_{i,j}\sigma_{ij}^\alpha \leq \frac{1}{\alpha}\sum_{k=2}^{n}\frac{1}{2\lambda_k} +\\
&+\frac{1}{\alpha}\sum_{k=2}^n\int_{0}^\infty  \int_{\alpha}^2 \| q_k\|_{w}^{2w} g^{w-1}(s)\big|\ln g(s)\big|e^{-\lambda_k s}\,dwds
\end{split}
\end{equation*} 
The result follows in view of the norm estimate $$\|q_k\|_{w}\leq n^{\frac{1}{w}-\frac{1}{2}}\|q_k\|_2=n^{\frac{1}{w}-\frac{1}{2}}$$
\end{proof}
\noindent The upper bound above, although true for $\alpha \in (1,2]$, it is not expected to provide efficient estimates for values of $\alpha$ very far away from 2, due mainly to the way it was obtained.

Spectral based estimates such as these of Theorem \ref{thm: main1} and Theorem \ref{thm: main2}, could possibly be leveraged when developing optimal design algorithms, that reform the communication parameters into a network, more robust to the imposed noise. In order to verify the qualification of the estimates we must validate their efficiency on different network topologies. This is in part the subject of \S \ref{sect: examples}. 
\subsection{Connection with the $p^{th}$-moment, for $p<\alpha$.}
Theorem \ref{prop: ydist}  asserts that the output dynamics are stable vectors with the same stability parameter as the noise sources. Following Property 3 of Proposition \ref{prop: propstrv}, the distribution attains moments up to any $p<\alpha$, for $\alpha <2$. In particular, 
\begin{equation}\label{eq: hpnorm}
\mathbb E\big[ \|y\|_p^p \big]=c^p(\alpha,\beta,p) \|\sigma\|_p^p
\end{equation} 
When $y$ is Gaussian (i.e. $\alpha=2$), $\mathbb E\big[ \|y\|_p^p \big]$ exists for $p\leq \alpha$.  As Proposition \ref{prop: propstrv} explains, for the non-Gaussian range of $\alpha$, the $p^{th}$ moments diverge at $p=\alpha$. It is thus unreasonable to try to obtain $\mathcal H_2$-norm based measures of performance for heavy-tailed consensus systems. This is why cumulative scale $\Sigma_\alpha(\overline{y})$ may be regarded as extension of the classic input/output $\mathcal H_2$ performance measure. Indeed, for $\alpha=2$, and $p=2$, the statistics of $\overline{y}$ recover the well-known formula \eqref{eq: h2norm}.

We conclude this section with reporting the relation of $\mathbb E\big[ \|y\|_p^p \big]$ and $\Sigma_\alpha$, through the basic equivalence properties of Euclidean norms. Straightforward calculations yield for $p<a$,
\begin{equation*}
c^p  \sqrt[p / \alpha]{\Sigma_{\alpha}(\overline{y})} \leq \mathbb E[ \| \overline{y} \|_p^p ] \leq n^{1-\frac{p}{\alpha}} c^p \sqrt[p / \alpha]{\Sigma_{\alpha}(\overline{y})}
\end{equation*} where $c=c(\alpha,\beta,p)$ is the constant in Proposition \ref{prop: propstrv}. Interestingly enough, the double inequality becomes exact at $\alpha=2$ and in the limit $p\rightarrow 2^{-}$.

\section{Numerical Examples}\label{sect: examples} In this section, we discuss four examples related to output \eqref{eq: outputdynamics}. The first three, regard elementary network design problems.  Their objective is to demonstrate that the basic design strategies (addition/removal of links and re-weighting) are  critically affected by parameter $\alpha$ of input noise. The fourth example is a validation of the estimates in Theorems \ref{thm: main1} and \ref{thm: main2}. Our focus is on consensus systems driven by symmetric $\alpha$-stable noise (i.e. $\beta=0$ and $\mu=0$).

\begin{figure*}[h]
\begin{tikzpicture}
  [scale=.6,auto=left,every node/.style={circle,fill=blue!20}]
  \node (n2) at (8,12)  {2};
  \node (n4) at (5,12)  {4};
  \node (n5) at (3,13)  {5};
  \node (n1) at (9,10)  {1};
  \node (n6) at (8,14)  {6};
  \node (n3) at (12,9)  {3};
  \node (n8) at (5,5) [fill=none] {$\boldsymbol{\mathcal G_1}$};
  \foreach \from/\to in {n2/n3,n2/n6,n3/n5,n5/n2}
   \path (\from) edge  [bend left] (\to) [thick];
\foreach \from/\to in {n3/n6}  
\path (\from) edge  [bend right] (\to) [thick];
\foreach \from/\to in {n2/n4,n1/n2}  
   \draw (\from)--(\to) [thick] ;
   \foreach \from/\to in {n3/n1}
   \foreach \from/\to in {n1/n4}   
\draw (\from) edge [bend left](\to) [blue,dashed,thick];
   \foreach \from/\to in {n1/n3}   
\draw (\from) --(\to) [red,dashed,thick];
   \foreach \from/\to in {n3/n4}   
\draw (\from) edge [bend left](\to) [red,dashed,thick];
\end{tikzpicture}\hspace{0.3in}
\begin{tikzpicture}
  [scale=.6,auto=left,every node/.style={circle,fill=blue!20}]
  \node (n1) at (1,10) {1};
  \node (n2) at (3,10)  {2};
  \node (n3) at (5,10)  {3};
  \node (n4) at (7,10) {4};
  \node (n5) at (9,10)  {5};
  \node (n6) at (1,5)  {6};
  \node (n7) at (3,5)  {7};
  \node (n8) at (5,5)  {8};
  \node (n9) at (7,5)  {9};
  \node (n10) at(9,5)  {10};
  
  \node (n12) at (5,1) [fill=none] {$\boldsymbol{\mathcal G_2}$};
  \foreach \from/\to in {n1/n2,n1/n3,n1/n4}
  \path (\from) edge  [bend left] (\to)[thick];
\foreach \from/\to in {n1/n6,n1/n9}  
   \draw (\from)--(\to)[thick];
     \foreach \from/\to in {n2/n4}
  \path (\from) edge  [bend left] (\to)[blue, dashed,very thick];
\foreach \from/\to in {n2/n8,n2/n9,n2/n10}  
   \draw (\from)--(\to)[thick];
   \foreach \from/\to in {n2/n6}  
\draw (\from)--(\to)[green,dashed,very thick];
        \foreach \from/\to in {n3/n4}
  \path (\from) edge  [bend left] (\to)[thick];
  \foreach \from/\to in {n3/n7,n3/n8}  
   \draw (\from)--(\to) [thick];
     \foreach \from/\to in {n4/n7,n4/n8,n4/n9,n4/n10}  
   \draw (\from)--(\to) [thick];
        \foreach \from/\to in {n5/n6,n5/n8,n5/n9,n5/n10}  
   \draw (\from)--(\to)[thick];
  \foreach \from/\to in {n6/n8,n6/n9,n6/n10}
  \path (\from) edge  [bend right] (\to)[thick];   
    \foreach \from/\to in {n7/n9,n7/n10}
  \path (\from) edge  [bend right] (\to)[thick]; 
      \foreach \from/\to in {n8/n10}
  \path (\from) edge  [bend right] (\to)[red,dashed,very thick];  
\end{tikzpicture} \hspace{0.3in}
\begin{tikzpicture}
  [scale=.6,auto=left,every node/.style={circle,fill=blue!20}]
  \node (n1) at (1,10)  {1};
  \node (n2) at (4,10)  {2};
  \node (n5) at (7,10)  {5};
  \node (n3) at (4,8) {3};
  \node (n4) at (7,8)  {4};
  \foreach \from/\to in {n1/n2,n2/n3,n2/n5,n5/n4}  
  \draw (\from)--(\to)[very thick];
  \draw (n1) -- (n2) node [midway,fill=none] {\textbf{1}} ;
  \draw (n2) -- (n5) node [midway,fill=none] {$a_{25}$} ;
  \draw (n2) -- (n3) node [midway,fill=none] {$a_{23}$} ;
   \draw (n4) -- (n5) node [midway,fill=none] {\textbf{1}} ;
    \node (n12) at (5,1) [fill=none] {$\boldsymbol{\mathcal G_3}$};
\end{tikzpicture}\caption{The graph topologies of Examples \ref{exmpl: ex2}, \ref{exmpl: sparsification} and \ref{exmpl: reweighting}, respectively.  } 
\label{fig: graph}
\end{figure*}
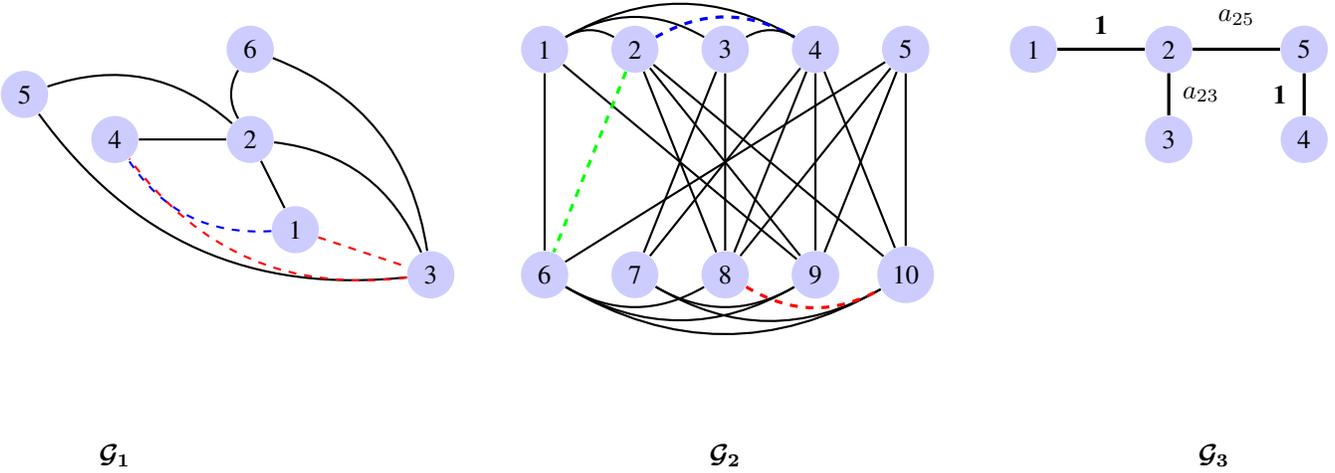 

\begin{exmp}\label{exmpl: ex2}[Design via Expansion] We consider a network over $n=6$ agents that seek consensus. The communication network, illustrated as $\boldsymbol{\mathcal G_1}$ in Figure \ref{fig: graph}, is a linear time-invariant with unit coupling links. The network is hit by stable noise forming the dynamics of \eqref{eq: model}. In this problem, we have the option to add a new unit-weight link to the network so as to improve its performance. In other words, we look for the link location, that upon establishing, $\Sigma_{\alpha}$ is minimized. Numerical explorations signify that the optimal selection is a function of $\alpha$.
For $\alpha=2$ to $\alpha=1.6655$ the optimal location is a link between nodes 1 and 4 (blue dotted curve). From $\alpha=1.6655$ to $\alpha=0.3312$ there appear to be two equivalent alternatives: one is the pair (1,3) and the other is (3,4) (red dashed curves). For stability values below $0.3313$ the optimal pair is (1,3).
\end{exmp}
\begin{exmp}\label{exmpl: sparsification}[Design via Sparsification] Next, we consider a dense linear network over 10 nodes. It is depicted as graph  $\boldsymbol{\mathcal G_2}$ in Figure \ref{fig: graph}. The working hypothesis is that the existence of too many links, makes for an expensive communication structure. The problem in this network is to choose the one link of the network that, upon removal, increases $\Sigma_\alpha$, the least. Our findings suggest that within the stability range $\alpha=2$ to $\alpha=1.8932$, the optimal pair is (2,4) (removal blue dashed curve). From $\alpha=1.8937$ to $\alpha=0.1971$ the optimal pair is (8,10) (removal of the red dashed curve). Finally, for $\alpha<0.1971$ the optimal pair appears to be (2,6) (removal of the green dashed curve).
\end{exmp}
\begin{exmp}\label{exmpl: reweighting}[Design via Re-weighting] In this last example, we regard a small network of 4 agents, illustrated as $\boldsymbol{\mathcal G_3}$ in Figure \ref{fig: graph}. All but links between nodes (2,3) and (2,5) are fixed and of unit weight. On the other hand, the edges $a_{23}$ and $a_{25}$ are assumed to satisfy $a_{23}=2-b, a_{25}=b$ for some $b\in (0,2) $. In other words, keeping the overall network budget constant and equal to $a_{12}+a_{23}+a_{25}+a_{45}=4$ we seek to calibrate the control parameter $b$ towards the value that minimizes $\Sigma_\alpha$. The simulations are illustrated in Figure \ref{fig: example3} where we essentially depict the dependence of the optimal calibration (the black dots) as a function of $\alpha$.
\begin{figure}[t]\center
\includegraphics[scale=0.6]{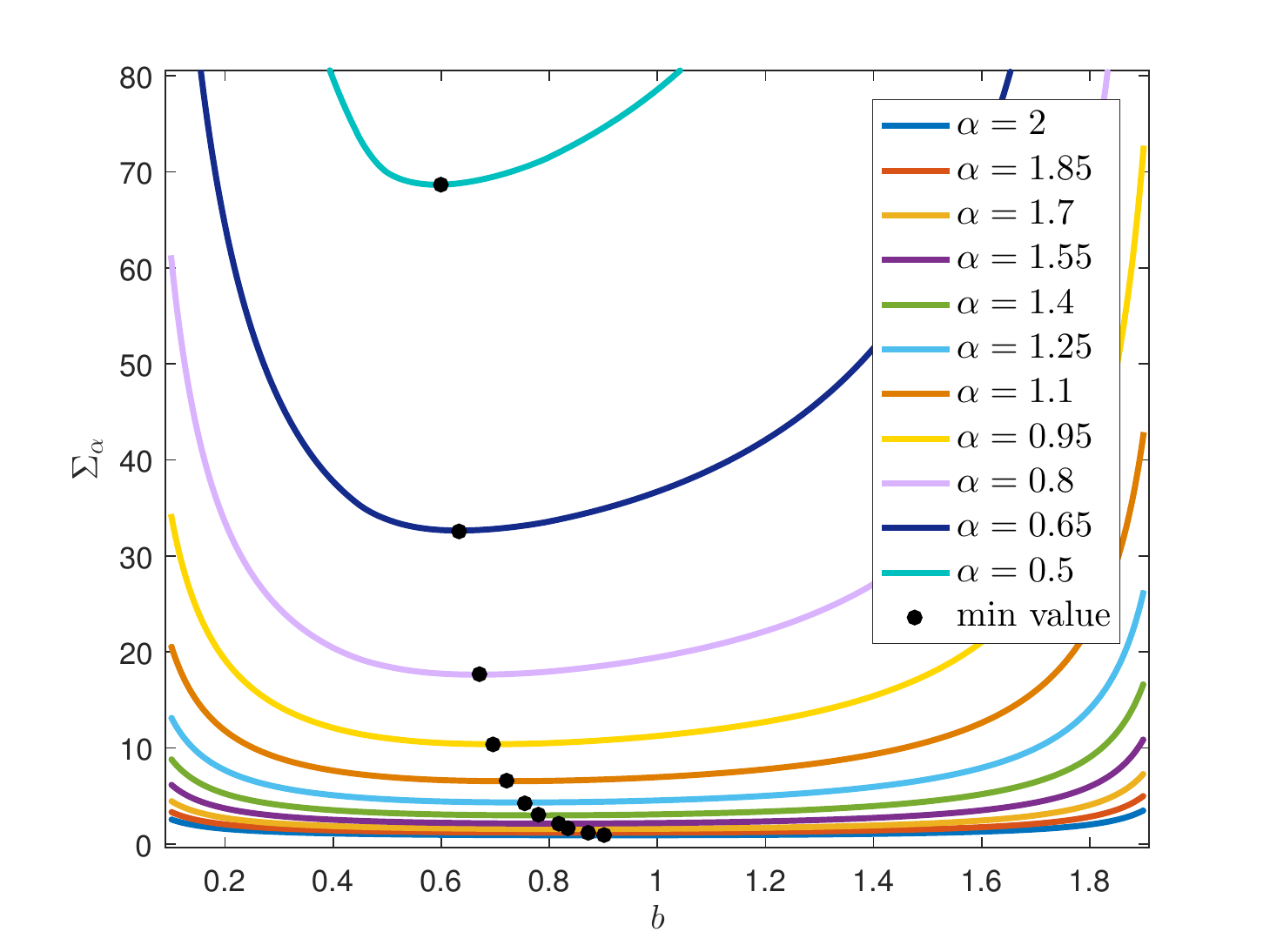}\caption{The cumulative scale parameter $\Sigma_{\alpha}$ as function of the control $b$, in Example \ref{exmpl: reweighting}. The different curves correspond to stable noises of various stability parameters. The lowest curve is this of the Gaussian case, $\alpha=2$. The $\Sigma_\alpha$ curves increase monotonically as $\alpha$ varies from 2 to 0, verifying Proposition \ref{prop: monotonicity}. The sequence of black dots signify the global minimum in each type of noise.}\label{fig: example3}
\end{figure}
\end{exmp}

All three network design problems lead to a definitive conclusion: performance evaluation tools that are associated with a particular type of stochastic uncertainty (e.g. the Gaussian and the associated $\mathcal H_2$ performance measure) become obsolete in other types of uncertainty (e.g. non-Gaussian cases).

\begin{exmp}\label{exmpl: estimates}
We test the scale estimates of Theorems \ref{thm: main1} and \ref{thm: main2}. We choose two graphs. The first graph has a significantly larger eigenvalue ratio than the second one. The curves are depicted in Figure \ref{fig: example3} and are compared with the exact value. There are generally two remarks due. The estimates perform better in graphs with ratio $\lambda_n/ \lambda_2$ close to 1. Also, as the noise distribution becomes more and more impulsive (smaller values of $\alpha$) the estimates becomes less and less efficient.
\begin{figure}[t]\center
\includegraphics[scale=0.54]{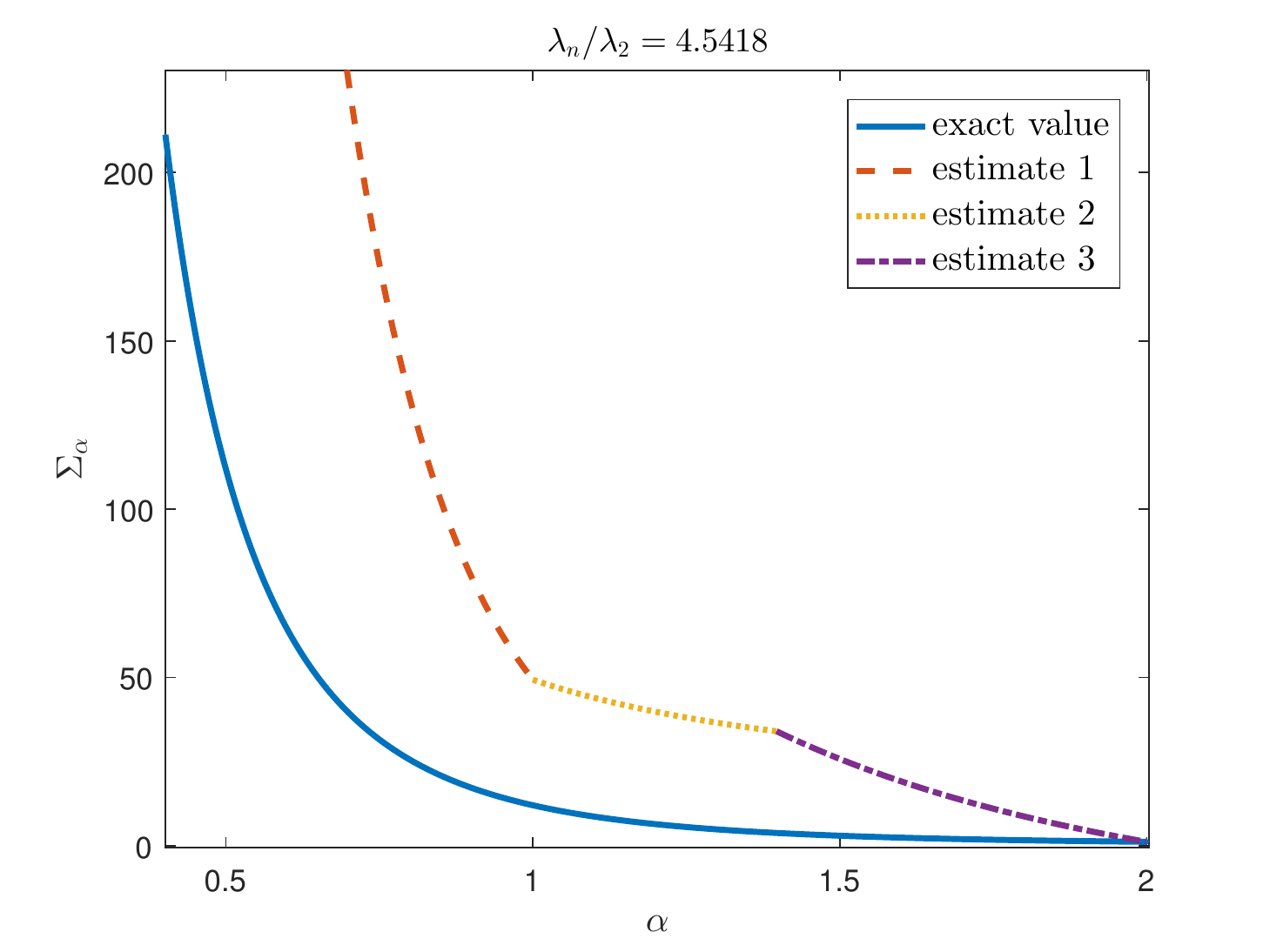} 
\includegraphics[scale=0.54]{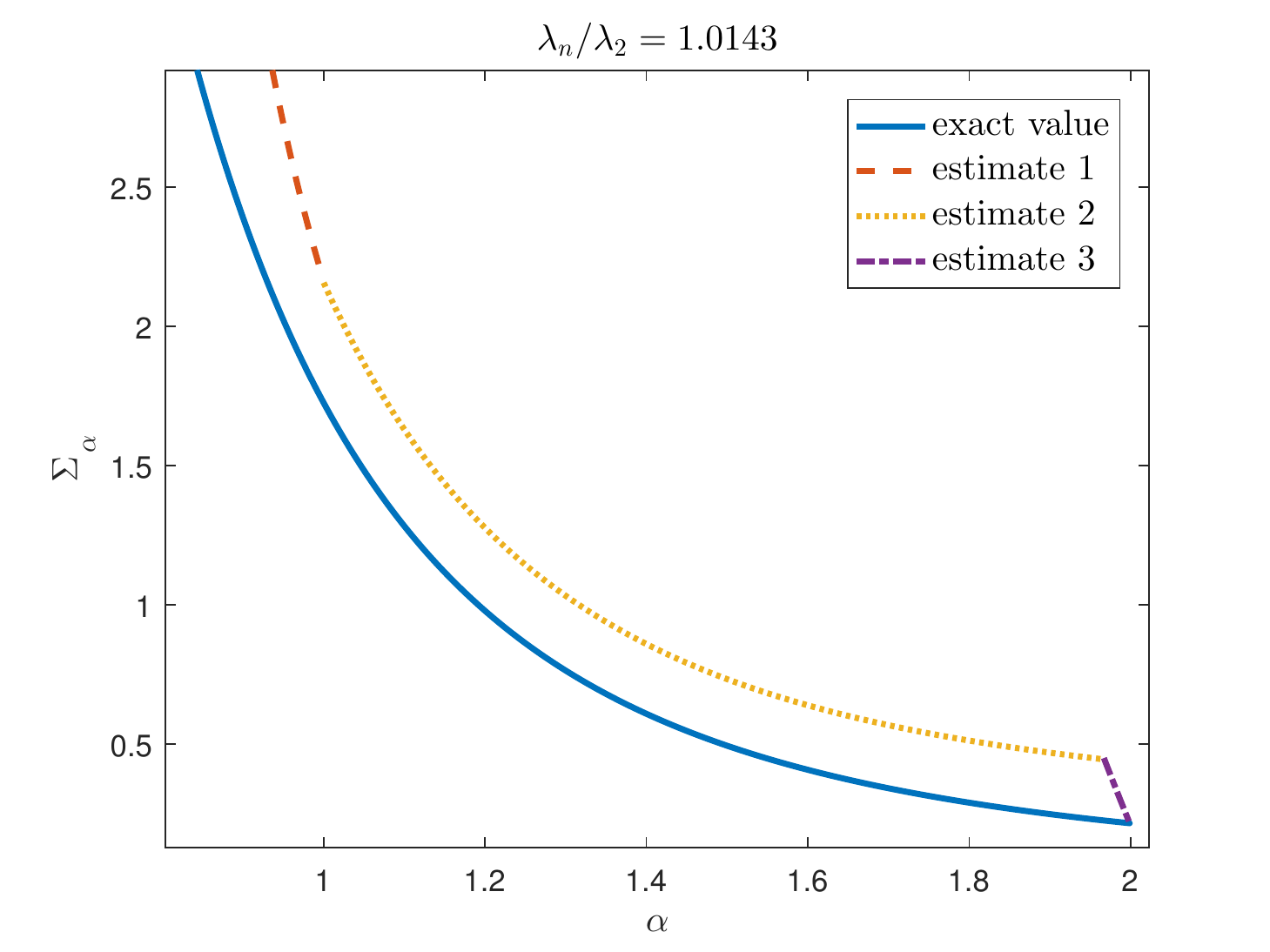}\caption{Simulation of Example \ref{exmpl: estimates}. Graphs with smaller $\frac{\lambda_n}{\lambda_2}$ ratio provide scale estimates closer to the actual value. Estimate 1, regards $\Sigma_\alpha$ with $\alpha\in (0,1]$. Estimate 2, regards $\Sigma_\alpha$ with $\alpha\in [1,2]$. Estimate 3, regards $\Sigma_\alpha$ with $\alpha\in [1,2]$ as in Theorem \ref{thm: main2}. }
\end{figure}
\end{exmp}

\section{DISCUSSION}

Modeling of uncertainty in networked control systems typically assumes noise sources generated by Brownian motion. Albeit popular, such perturbations are not rich enough to incorporate real-world uncertainties that incorporate impulsive shocks. In this paper, we considered consensus seeking systems in the presence of sources induced by heavy-tailed probability measures. 

We defined extensions of measures of performance that quantify systemic response in the presence of heavy-tailed noise. These were cumulative scale parameters of $\alpha$-stable vectors, that demonstrate close relations with the $p$-norms of output dynamics. It is argued that heavy-tailed performance measures may be regarded as generalization of $\mathcal H_2$-norm based measures of performance for linear systems with white noise inputs.  Unless certain types of networks or noise are assumed, explicity calculation of heavy-tailed performance measures is not possible. Our estimates perform quite well for types of networks with the property that the graph laplacian eigenvalues satisfy $\lambda_n/\lambda_2 \approx 1$. In addition to complete graph connectivity (where  $\lambda_n/\lambda_2 = 1$) expander graphs also satisfy ratio $\lambda_n/\lambda_2 $, \cite{spielman}. Finally, we presented simple network design examples on $\alpha$-stable consensus network where we demonstrate that any optimal synthesis strategy must take into account the shock-impulsiveness of infused noise.  

\appendices

\section*{Appendix}

We proceed with reviewing some fundamental inequalities related to the function $s(t)=|t|^p$ for $p>0$. These inequalities play a crucial role in the derivation of the technical results of our paper.

\vspace{0.1in}
\begin{lemma}\label{lem: ineq} Let $u,v\in \mathbb R$. If $0<p\leq 1$, then  $$|u+v|^{p}\leq |v|^p+|u|^p.$$ If $p\in (1,2)$, then $$|u+v|^p\leq \min\big\{ 2^{p-1}( |u|^p+|v|^p ), |u|^p+|v|^p +p\,|v|^{p-1}\,|u|^p\big\}.$$
\end{lemma}
\begin{proof} For the first inequality as well as $|u+v|^p \leq 2^{p-1}(|v|^p+|u|^p)$ for $p>1$, we refer to \cite{samorodnitsky1994stable}. It remains to show that for $p\in (1,2]$ $|v+u|^p\leq |v|^p+|u|^p+p |v|^{p-1} |u|^p$. For this we write
\begin{equation*}
\begin{split}
|u+v|^p&=|v|^p+|u+v|^p-|v|^p\\
&\leq  |v|^p+p |u| \int_0^1 |q(u+v)+(1-q)v|^{p-1}\,dq \\
&\leq |v|^p+p |u| \int_0^1 |q u|^{p-1}\,dq+ p |u| |v|^{p-1}
\end{split}
\end{equation*} where the last step is due to the first inequality.
\end{proof}
\vspace{0.1in}
The estimate of $|\cdot|^p$ for $p \in [1,2)$ relies on two inequalities. The first one coincides with the inequality on $p \in (0,1]$, providing sharper estimates. The second inequality becomes exact if and only if either $u$ or $v$ is zero.

\vspace{0.1in}
\begin{lemma}\label{lem: jensen}
Let $\phi$ be a positive homogeneous of degree $p>1$ and convex function, defined on $\mathbb R$. Let real numbers $y_1,\dots,y_m$ and $b_1,\dots,b_m$ non-negative with $\sum_i b_i>0$. Then 
\begin{equation*}
\phi\bigg(\sum_{i=1}^m  b_i y_i\bigg)\leq \bigg(\sum_{i=1}^m b_i\bigg)^{p-1} \sum_{i=1}^m b_i\phi(y_i)
\end{equation*}
\end{lemma}
\begin{proof}
We write
\begin{equation*}
\begin{split}
\phi\bigg(\sum_{i=1}^m  b_i y_i\bigg)&=\phi\bigg(w \cdot \frac{\sum_{i}  b_i y_i}{w}\bigg)=w^{p}\phi\bigg(\frac{\sum_{i}  b_i y_i}{w}\bigg)
\end{split}
\end{equation*} 
where $w=\sum_{i} b_i>0$. The result follows by direct application of Jensen's inequality \cite{royden} on $\phi\big(\frac{\sum_{i}  b_i y_i}{\sum_{i} b_i}\big)$.
\end{proof}

\vspace{0.1in}
\begin{lemma}[Minkowski's Inequality \cite{kolmogorov1975introductory}] \label{lem: minksowski} Let $p\geq 1$ and $f,~g$ real-valued, integrable functions on $E\subset \mathbb R$. Then 
$$\bigg(\int_E |f+g|^p\,ds\bigg)^{\frac{1}{p}}\leq \bigg(\int_E |f|^p\,ds\bigg)^{\frac{1}{p}}+ \bigg(\int_E |g|^p\,ds\bigg)^{\frac{1}{p}}.  $$
\end{lemma}

\begin{proof}[Proof of Claim \ref{lem: sigmaijestimates}] Observe that Assumption \ref{assum0} implies, \begin{equation*}\sum_{k=2}^n q_{ik}q_{jk}= \begin{cases} -\frac{1}{n}, & i\neq j\\
\frac{n-1}{n}, & i=j.
\end{cases}
\end{equation*}Based on this property and elementary algebra we observe that $f_{ij}(t)$ can be re-written as:
\begin{equation*}
f_{ij}(t)=\begin{cases} W_{i,j,n}(t)-\frac{1}{n(n-1)}g(t), & i\neq j \\
W_{i,j,n}(t)+\frac{1}{n}g(t), & i=j
\end{cases}
\end{equation*} where $$W_{i,j,n}(t)=\frac{1}{n-1}\sum_{k=2}^n\sum_{m\neq 1,k}(e^{-\lambda_k t}-e^{-\lambda_m t})q_{ik}q_{jk},$$ and $g(t)=\sum_{k=2}^n e^{-\lambda_k t}$.
We elaborate only for $i\neq j$. Repeated application of Lemma \ref{lem: ineq}, followed by $e^{-x}\geq 1-x$ gives

\begin{equation*}
\begin{split}
&|f_{ij}(t)|^{\alpha}\leq \frac{1}{(n-1)^{\alpha}}\sum_{k=2}^{n}|q_{ik}|^\alpha |q_{jk}|^{\alpha}\times \\
&\bigg[ \sum_{m=2}^{k=1}e^{-\lambda_m t}\big(1-e^{-(\lambda_k-\lambda_m)t}\big) +\\
&\hspace{0.4in}+e^{-\lambda_k t}\sum_{m=k+1}^{n}\big(1-e^{-(\lambda_m-\lambda_k)t}\big) \bigg]^{\alpha}+\frac{g^{\alpha}(t)}{n^{\alpha}(n-1)^{\alpha}}  
\end{split}
\end{equation*} 
and 

\begin{equation*}\begin{split}
|f_{ij}(t)|^{\alpha}&\leq \frac{1}{(n-1)^{\alpha}}\sum_{k=2}^{n}|q_{ik}|^\alpha |q_{jk}|^{\alpha}\times\\&\bigg[\sum_{m=2}^{k-1}e^{-\alpha\lambda_m t}(\lambda_k-\lambda_m)^{\alpha}t^{\alpha}+\\
&+e^{-\alpha\lambda_k t}\sum_{m=k+1}^{n}\big(\lambda_m-\lambda_k\big)^\alpha t^\alpha\bigg]+\frac{g^{\alpha}(t)}{n^{\alpha}(n-1)^{\alpha}}
\end{split}
\end{equation*}

Consequently,

\begin{equation*}
\begin{split}
&\int_0^{\infty}|f_{ij}(s)|^{\alpha}\,ds\leq \frac{1}{(n-1)^{\alpha}}\sum_{k=2}^{n}|q_{ik}|^\alpha |q_{jk}|^{\alpha}\times\\
&\Gamma(\alpha+1)\bigg(\sum_{m=2}^{k-1}\frac{(\lambda_k-\lambda_m)^{\alpha}}{(\alpha \lambda_m)^{\alpha}}+\frac{1}{(a\lambda_k)^{\alpha}}\sum_{m=k+1}^{n}\big(\lambda_m-\lambda_k\big)^\alpha\bigg)\\
&+\frac{\int_0^\infty g^\alpha(s)\,ds}{n^\alpha(n-1)^\alpha}=\frac{1}{(n-1)^\alpha}\bigg[\sum_{k=2}^n |q_{ik}|^\alpha |q_{jk}|^\alpha \Lambda_{\alpha,1}(k)+\frac{G_\alpha}{n^\alpha}\bigg]
\end{split}
\end{equation*}

Following similar steps, for $i=j$, we have 

\begin{equation*}
\begin{split}
&\int_0^{\infty}|f_{ij}(s)|^{\alpha}\,ds\leq\\
&\hspace{0.25in}\leq \frac{1}{(n-1)^\alpha}\bigg[\sum_{k=2}^n |q_{ik}|^\alpha |q_{jk}|^\alpha \lambda_{\alpha,1}(k)+\frac{(n-1)^\alpha}{n^\alpha} G_\alpha\bigg].
\end{split}
\end{equation*}
\end{proof}

\begin{proof}[Proof of Claim \ref{lem: sigmaijestimates2}] For $\alpha \in [1,2]$, we invoke Lemma \ref{lem: minksowski} and use similar techniques to these in proof of Claim \ref{lem: sigmaijestimates} to obtain
\begin{equation*}
\sigma_{ij}\leq \begin{cases} \frac{1}{\alpha^{\frac{1}{\alpha}}(n-1)}\bigg[\sum_{k} |q_{ik}| |q_{jk}|\Lambda_{\alpha,\alpha}(k)+\frac{1}{n}G_{\alpha}^{\frac{1}{\alpha}}\bigg], & i\neq j \\
\frac{1}{\alpha^{\frac{1}{\alpha}}(n-1)}\sum_{k} |q_{ik}|^2\Lambda_{\alpha,\alpha}(k)+\frac{1}{\alpha^{\frac{1}{\alpha}} n}G_{\alpha}^{\frac{1}{\alpha}}, & i=j
\end{cases} 
\end{equation*} 
We need, however, estimates of $\sigma_{ij}^\alpha$. So
\begin{equation*}
\sigma_{ij}^\alpha\leq \begin{cases} \frac{1}{\alpha(n-1)^\alpha}\big[\sum_{k} |q_{ik}| |q_{jk}|\Lambda_{\alpha,\alpha}(k)+\frac{1}{n}G_{\alpha}^{\frac{1}{\alpha}}\big]^\alpha, & i\neq j \\
\frac{1}{\alpha(n-1)^\alpha}\big[\sum_{k} |q_{ik}|^2\Lambda_{\alpha,\alpha}(k)+\frac{(n-1)}{ n}G_{\alpha}^{\frac{1}{\alpha}}\big]^\alpha, & i=j
\end{cases} 
\end{equation*} 
Now, we use the second inequality of Lemma \ref{lem: ineq}, part by part. Fix $i\neq j$. Then, either
\begin{equation*}\begin{split}
\sigma_{ij}^\alpha&\leq \frac{2^{\alpha-1}}{\alpha(n-1)^\alpha}\bigg[\bigg(\sum_k |q_{ik}| |q_{jk}|\Lambda_{\alpha,\alpha}(k) \bigg)^\alpha+\frac{G_\alpha}{n^\alpha}\bigg]\\
& \leq \frac{2^{\alpha-1}}{\alpha(n-1)^\alpha}\bigg[\Lambda_{\alpha,\alpha}^{\alpha-1}\sum_k |q_{ik}|^\alpha |q_{jk}|^\alpha \Lambda_{\alpha,\alpha}(k)+\frac{G_\alpha}{n^\alpha}\bigg],
\end{split}
\end{equation*} where the last step is due to Lemma \ref{lem: jensen}, or
\begin{equation*}\begin{split}
\sigma_{ij}^\alpha&\leq \frac{1}{\alpha(n-1)^\alpha}\bigg[\bigg(\sum_k |q_{ik}| |q_{jk}|\Lambda_{\alpha,\alpha}(k) \bigg)^\alpha+\\
&\hspace{0.7in}+\frac{G_\alpha}{n^\alpha}+\alpha \big(\sum_k |q_{ik}| |q_{jk}|\Lambda_{\alpha,\alpha}(k)\big)^{\alpha-1}\frac{G_\alpha}{n^\alpha} \bigg]\\
& \leq \frac{1}{\alpha(n-1)^\alpha}\bigg[\Lambda_{\alpha,\alpha}^{\alpha-1}\sum_k |q_{ik}|^\alpha |q_{jk}|^\alpha \Lambda_{\alpha,\alpha}(k)\\
&\hspace{1.9in}+\frac{G_\alpha}{n^\alpha}\big(1+\alpha \Lambda_{\alpha,\alpha}^{\alpha-1}\big)\bigg]
\end{split}
\end{equation*} for the last step is due $(\sum_{k}|q_{ik}| |q_{jk}| \Lambda_{\alpha,\alpha}(k))^{\alpha-1}\leq \Lambda_{\alpha,\alpha}^{\alpha-1}$.
Similar steps are taken for $i=j$. 
\end{proof}

%

\ifCLASSOPTIONcaptionsoff
  \newpage
\fi



\bibliographystyle{IEEEtran}
\bibliography{bibliography}

\end{document}